\documentclass[11pt]{article}
\usepackage[utf8]{inputenc}

\usepackage{fullpage}
\usepackage{amsmath,amssymb,amsfonts,amsthm,xspace,hyperref,graphicx,relsize,bm}
\usepackage{mathpazo}
\usepackage[margin=1in]{geometry}
\usepackage{enumitem}
\usepackage{xcolor}
\usepackage{cleveref}
\usepackage{stmaryrd}
\usepackage{dsfont}
\usepackage{thm-restate}
\usepackage{mathtools}

\usepackage{url}
\usepackage{upgreek}

\usepackage[linesnumbered,ruled,vlined]{algorithm2e}

\newcommand{\suppress}[1]{}

\newtheorem{theorem}{Theorem}[section]

\newtheorem{lemma}[theorem]{Lemma}
\newtheorem{corollary}[theorem]{Corollary}
\newtheorem{definition}[theorem]{Definition}
\theoremstyle{definition}

\numberwithin{equation}{section}

\newcommand{\integers}{{\mathbb Z}}

\newcommand{\union}{\cup}
\newcommand{\intersect}{\cap}

\newcommand{\eqdef}{\coloneqq}

\newcommand{\size}[1]{\left| #1 \right|}
\newcommand{\set}[1]{\left\{ #1 \right\}}
\newcommand{\floor}[1]{{\lfloor #1 \rfloor}}
\newcommand{\ceil}[1]{{\lceil #1 \rceil}}

\newcommand{\Order}{\mathrm{O}}
\newcommand{\order}{\mathrm{o}}

\newcommand{\bR}{{\bm R}}
\newcommand{\bh}{{\bm h}}

\title{Deterministic Algorithms for the Hidden Subgroup Problem}

\author{Ashwin Nayak~\thanks{Department of Combinatorics and Optimization,
and Institute for Quantum Computing, 
University of Waterloo, 200 University Ave.\ W., Waterloo, ON,
N2L~3G1, Canada. Email: \texttt{ashwin.nayak@uwaterloo.ca}~.} \\
University of Waterloo
}

\date{March 15, 2022}

\begin{document}

\maketitle

\begin{abstract}
We present deterministic algorithms for the Hidden Subgroup Problem. The first algorithm, for abelian groups, achieves the same asymptotic worst-case query complexity as the optimal randomized algorithm, namely~$ \Order(\sqrt{ n}\, )$, where~$n$ is the order of the group. The analogous algorithm for non-abelian groups comes within a~$\sqrt{ \log n}$ factor of the optimal randomized query complexity.

The best known randomized algorithm for the Hidden Subgroup Problem has \emph{expected\/} query complexity that is sensitive to the input, namely~$ \Order(\sqrt{ n/m}\, )$, where~$m$ is the order of the hidden subgroup. In the first version of this article~\cite[Sec.~5]{Nayak21-hsp-classical}, we asked if there is a deterministic algorithm whose query complexity has a similar dependence on the order of the hidden subgroup. Prompted by this question, Ye and Li~\cite{YL21-hsp-classical} present deterministic algorithms for \emph{abelian\/} groups which solve the problem with~$ \Order(\sqrt{ n/m }\, )$ queries, and find the hidden subgroup with~$ \Order( \sqrt{ n (\log m)  / m} + \log m ) $ queries. Moreover, they exhibit instances which show that in general, the deterministic query complexity of the problem may be~$\order(\sqrt{ n/m } \,)$, and that of \emph{finding\/} the entire subgroup may also be~$\order(\sqrt{ n/m } \,)$ or even~$\upomega(\sqrt{ n/m } \,) $.

We present a different deterministic algorithm for the Hidden Subgroup Problem that also has query complexity~$ \Order(\sqrt{ n/m }\, )$ for abelian groups. The algorithm is arguably simpler. Moreover, it works for non-abelian groups, and has query complexity~$ \Order(\sqrt{ (n/m) \log (n/m) }\,) $ for a large class of instances, such as those over supersolvable groups. We build on this to design deterministic algorithms to find the hidden subgroup for all abelian and some non-abelian instances, at the cost of a~$\log m$ multiplicative factor increase in the query complexity.
\end{abstract}

\section{Introduction}

In the Simon Problem with parameter~$k$, we are given an oracle for a function~$f : \integers_2^k \rightarrow S$, for some co-domain~$S$ with~$\size{S} \ge 2^k$. The function~$f$ is either injective, or there is an unknown non-zero element~$s \in \integers_2^k$ such that for all~$x,y \in \integers_2^k$, we have~$f(x) = f(y)$ if and only if~$x + y \in \set{0, s}$. In the latter case, we say that the function ``hides'' the element~$s$, and call a pair of distinct inputs~$x,y$ a ``collision'' if~$f(x) = f(y)$. The task is to determine which of the two cases holds. 

Simon designed a bounded-error quantum algorithm with query complexity~$\Order(k)$ for the eponymous problem, and showed that any bounded-error classical (i.e., randomized) algorithm for the problem requires~$\Omega(2^{k/2})$ queries~\cite{S97-simon-algorithm}. (The query lower bound stated in the paper is $\Omega(2^{k/4})$ for classical algorithms with error at most~$\tfrac{1}{2} - 2^{-k/2}$. However, the proof can be modified in a straightforward manner to show that if the algorithm makes error at most~$\tfrac{1}{4}$, at least~$2^{k/2 - 1}$ queries are required.) The lower bound for classical algorithms is optimal up to a constant factor, as is implied by the ``Birthday Paradox'': if we pick~$t \eqdef \ceil{2^{k/2 + 1}}$ elements~$X_1, X_2, \dotsc, X_t$ independently and uniformly at random from~$\integers_2^k$, if~$f$ hides some non-zero element, with probability at least~$3/4$ we find a collision (i.e., a pair~$X_i, X_j$ with~$i,j \in [t]$ such that~$X_i \neq X_j$ but~$f(X_i) = f(X_j)$). 

We present a simple, \emph{deterministic\/} algorithm for the Simon problem that achieves the asymptotically optimal classical query complexity (Theorem~\ref{thm-simon-det}). Since posting the first version of this article~\cite{Nayak21-hsp-classical}, we have learnt that algorithms achieving the same asymptotic query complexity, including the same algorithm, were known before~\cite{CQ18-simon-problem,WQTLC21-generalized-simon,YHLW21-generalized-simon}. Nevertheless, it is instructive to understand the algorithm underlying Theorem~\ref{thm-simon-det}, as it forms the basis of the generalizations that we describe next.

The Simon Problem is a special case of the Hidden Subgroup Problem (see, e.g., Ref.~\cite{CvD10-quantum-algorithms}). In the Hidden Subgroup Problem, we are given the description of a finite group~$G$, and an oracle for a  function~$f : G \rightarrow S$, where~$\size{S} \ge \size{G}$. The function~$f$ is either injective, or there is an unknown non-trivial subgroup~$H$ of~$G$ such that for all~$x,y \in G$, we have~$f(x) = f(y)$ if and only if~$x^{-1} y \in H$. In other words, the function~$f$ is constant on left cosets of a possibly trivial subgroup of~$G$, and takes distinct values for distinct left cosets of the subgroup. We say that the function ``hides'' the subgroup~$H$. The task is to determine whether the subgroup is trivial, i.e., the function~$f$ is injective, or not. When~$G \eqdef \integers_2^k$ and the hidden subgroup~$H$ is of the form~$\set{0, s}$ for some unknown non-zero element~$s \in \integers_2^k$, we get the Simon Problem.

The randomized algorithm for the Simon Problem generalizes immediately to the Hidden Subgroup Problem when the order~$m$ of the hidden subgroup~$H$ is known, and the resulting algorithm has query complexity~$\Order(\sqrt{ n/m }\, )$, where~$n$ is the order of the group~$G$. When the order~$m$ is not known, we can use this basic algorithm to design a new algorithm. The new algorithm has query complexity~$\Order(\sqrt{n} \, )$ when the function is injective, and \emph{expected\/} query complexity~$\Order(\sqrt{ n/m }\, )$ when the function~$f$ hides some non-trivial subgroup~$H$ with unknown order~$m$. Namely, we run the basic randomized algorithm above, with error at most~$1/4$, assuming that the order of the hidden subgroup is at least~$r$, starting from~$r \eqdef n/2$. If we do not succeed, we halve~$r$, and repeat (until~$r \le 1$).

We may ask if the Hidden Subgroup Problem also admits a deterministic algorithm that is as efficient in the worst-case as the best randomized algorithm. We answer this in the affirmative when the underlying group is abelian (\Cref{cor-ahsp}). For non-abelian groups, we present a deterministic algorithm that comes within a multiplicative factor of $\Order(\sqrt{\log n} \,)$ of the optimal randomized query complexity (\Cref{cor-non-abelian-hsp}). Both the algorithms are based on the construction of a \emph{generating pair\/} of subsets for a group (\Cref{def-gen-pair}), with optimal or nearly optimal size.

The optimal randomized algorithm for the Hidden Subgroup Problem described above has expected query complexity that is sensitive to the input, namely~$ \Order(\sqrt{ n/m}\, )$. It is natural to ask if there is a deterministic algorithm that has query complexity with a similar dependence on the order of the hidden subgroup. Algorithms coming close to this bound were known prior to this work for the case of~$G \eqdef \integers_p^k$ for prime~$p$, when the subgroup~$H$ has order~$p^l$ and~$l$ is known; see Ref.~\cite{WQTLC21-generalized-simon} for the~$p = 2$ case, and Ref.~\cite{YHLW21-generalized-simon} for the general case. The algorithm in Ref.~\cite{YHLW21-generalized-simon} has query complexity~$ \Order( \sqrt{ n (\log m)  / m} + \log m ) $.
\suppress{
of the order of
\[
\max\set{\log m, ~ \sqrt{ n (\log m) / m } \, } \enspace.
\]
}

The above question was posed as an open problem in the first version of this article~\cite[Sec.~5]{Nayak21-hsp-classical}. Prompted by this question, Ye and Li~\cite{YL21-hsp-classical} present deterministic algorithms for \emph{abelian\/} groups which solve the problem with query complexity~$ \Order(\sqrt{ n/m }\, )$, and find the hidden subgroup with query complexity~$ \Order( \sqrt{ n (\log m)  / m} + \log m ) $. These algorithms also build on the concept of generating pairs which we introduced in Ref.~\cite{Nayak21-hsp-classical}, and use a construction of generating pairs similar to the one we gave. Moreover, Ye and Li exhibit instances which show that in general, the deterministic query complexity of the problem may be~$\order(\sqrt{ n/m } \,)$, and that of \emph{finding\/} the entire subgroup may also be~$\order(\sqrt{ n/m } \,)$ or even~$\upomega(\sqrt{ n/m } \,) $. In fact, the instance with complexity~$\upomega(\sqrt{ n/m } \,) $ for finding the hidden subgroup follows from earlier work due to Ye, Huang, Li, and Wang~\cite[Theorem~1]{YHLW21-generalized-simon}. In more detail, for the group~$\integers_{p^k}$ for a given prime~$p$ and~$k \ge 2$, the deterministic query complexity of the problem is~$2$ and that of finding the hidden subgroup is~$\Order( \log (n/m))$, while~$ n/m $ may be~$ \upomega(1)$.  For~$G \eqdef \integers_p^k$ and a hidden subgroup~$H$ of order~$p^{k-1}$, any deterministic algorithm that \emph{finds\/} the subgroup has query complexity~$\Omega(k)$, while~$\sqrt{ n/m }$ is~$\sqrt{p}$. This gives an arbitrarily large separation as~$k$ grows for a fixed prime~$p$. 

We present a different deterministic algorithm for the Hidden Subgroup Problem that also has query complexity~$ \Order(\sqrt{ n/m }\, )$ for abelian groups (Algorithm~\ref{alg-fc}, \Cref{thm-fc}). The algorithm is arguably simpler. Furthermore, it also works for non-abelian groups, and has query complexity~$ \Order(\sqrt{ (n/m) \log (n/m) }\,) $ for classes of instances which include those over supersolvable groups. We observe that for abelian groups the problem of \emph{finding\/} the hidden subgroup may be reduced to that of finding a single collision, and obtain an~$ \Order \big( (\log m) \sqrt{ n/m }\, \big)$-query deterministic algorithm for it (Algorithm~\ref{alg-find-ahs}, \Cref{thm-abelian-find-hs}). We build on these results to design an algorithm that finds the hidden subgroup in certain non-abelian instances with~$ \Order \big( (\log m) \sqrt{ (n/m) \log (n/m)  }\, \big) $
\suppress{
\[
\Order \! \left( (\log m) \sqrt{ (n/m) \log (n/m)  }\, \right)
\]
}
queries (Algorithm~\ref{alg-fnc}, \Cref{thm-find-hs}). The instances are precisely the ones in which the underlying group is a bicrossed product of the hidden subgroup with another subgroup~\cite{ACIM09-bicrossed-product}. All these algorithms again rest on generating pairs of (near-) optimal size.

Note that the query complexity of the algorithm due to Ye and Li for abelian groups may be a factor of~$\sqrt{\log m}$ smaller than that of Algorithm~\ref{alg-find-ahs}. They achieve the stronger bound by searching for a structured set of independent generators for the hidden subgroup, and by using a partially nested sequence of generating pairs with finely tuned size. This entails a detailed analysis of the structure of the hidden subgroup. It is not clear whether we can achieve the same query complexity without resorting to these ideas.

\paragraph{Acknowledgements.}
A.N.\ is grateful for the opportunity to teach quantum computation at the undergraduate level, which prompted this work. He is also grateful to Kanstantsin Pashkovich for several helpful discussions, especially for a course-correction early in this work. He thanks William Slofstra for a pointer to relevant literature, and Zekun Ye for bringing Ref.~\cite{YHLW21-generalized-simon} to his attention. This research is supported in part by a Discovery Grant from NSERC Canada.

\section{The Simon Problem}
\label{sec-simon}

We start with a simple, deterministic algorithm for the Simon problem that matches the query complexity of the asymptotically optimal classical algorithm.
\begin{theorem}
\label{thm-simon-det}
There is a deterministic algorithm that makes~$2^{\floor{k/2}} + 2^{\ceil{k/2}}$ queries and solves the Simon Problem with parameter~$k$. Moreover, if the input function hides the non-zero element~$s$, the algorithm finds~$s$. Finally, any deterministic algorithm for the problem requires at least~$q$ queries, where~$q$ is the smallest positive integer such that~$\binom{q}{2} \ge 2^k - 1$. 
\end{theorem}
\begin{proof}
Let~$l \eqdef \floor{k/2}$. Viewing elements of~$\integers_2^k$ as~$k$-bit strings, we query the function~$f$ at all~$2^l$ elements of the form~$u \, 0^{k - l}$, where~$u \in \integers_2^l$, and at all~$2^{k-l}$ elements of the form~$0^l \, v$, where~$v \in \integers_2^{k-l}$. If the function is distinct at all these points, we say~$f$ is injective. Otherwise, we say that~$f$ is not injective, and output~$x + y$, where~$x \neq y$ and~$x,y$ are a colliding pair (i.e., are such that~$f(x) = f(y)$).

If the function is injective, the above algorithm outputs the correct answer. Suppose~$f$ hides a non-zero element~$s \in \integers_2^k$. Let~$a$ be the projection of~$s$ onto the first~$l$ coordinates, and~$b$ be the projection of~$s$ onto the last~$k-l$ coordinates. Then~$a \, 0^{k-l} + 0^l \, b = s$. As~$f(a \, 0^{k-l}) = f(0^l \, b)$, the above algorithm detects a collision  and computes~$s$ correctly.

Now consider any deterministic algorithm for the problem that makes~$t$ queries such that
\[
\binom{t}{2} \quad \le \quad 2^k - 2 \enspace.
\]
Then we argue that there is an injective function~$f_0$, and a function~$f_1$ that hides a non-zero element~$s$, such that~$f_0$ and~$f_1$ agree on all the~$t$ elements queried. Suppose the~$t$ queries that the algorithm makes are~$x_1, x_2, \dotsc, x_t$. We take~$f_0$ to be any injective function. Consider the set of elements~$S \eqdef \set{ x_i + x_j : i, j \in [t], i \neq j}$. We have~$\size{S} \le \binom{t}{2} \le 2^k - 2$, so there is at least one non-zero element in~$\integers_2^k \setminus S$. Let~$s$ be one such element. By definition of~$s$, we have~$x_i + s \neq x_j$ for any~$i,j \in [t]$. Thus there is a function~$f_1$ that equals~$f_0$ at all the points~$x_i$, $i \in [t]$, and also hides~$s$. 
\end{proof}
Note that for~$k \ge 2$, if~$\binom{q}{2} \ge 2^k - 1$, we have~$q^2 \ge 2^{k+1} - 2$, and therefore
\[
q \quad \ge \quad 2^{(k+1)/2} (1 - 1/2^k)^{1/2} \quad \ge  \quad 2^{(k+1)/2} - 2^{- (k-1)/2} \enspace.
\]
Since~$q$ is integral, $q \ge 2^{(k+1)/2}$. The deterministic upper bound in the theorem comes within a factor of~$3/2$ of this lower bound. Also note that the upper bound in the theorem is within a factor of~$3\sqrt{2}$ of the query lower bound for randomized algorithms with error at most~$1/4$.

\section{The Abelian Hidden Subgroup Problem}
\label{sec-abelian}

We now turn our focus to the Hidden Subgroup Problem when the underlying group is abelian. It is not clear how the use of projections in the algorithm presented in Section~\ref{sec-simon} may be extended to this case, let alone to the non-abelian case. The issue is that the given abelian group may be the direct product of cyclic groups with vastly different orders. We give a different extension of the algorithm, which combines several ways of expressing an abelian group as a sum of two subsets.

More formally, the basic idea behind the algorithm in Theorem~\ref{thm-simon-det} is to find a pair of sets~$S_1, S_2$ of elements of the group~$G$ such that~$\size{S_1}, \size{S_2} \in \Order(\sqrt{\size{G}} \, )$ and~$S_1 + S_2 = G$, where
\[
S_1 + S_2 \quad \eqdef \quad \set{x + y : x \in S_1, ~ y \in S_2} \enspace. 
\]
(Here, `$+$' denotes the group operation.) We explain how such a pair of sets may be constructed in a few cases. Together, they yield a construction for a general abelian group.

\begin{definition}
\label{def-gen-pair}
For any group~$G$, we say a pair~$S_1, S_2 \subseteq G$ is a \emph{generating pair\/} for~$G$ if the set~$S_1 S_2$ defined as~$S_1 S_2  \eqdef \set{ xy : x \in S_1, ~ y \in S_2 } $ equals~$G$.
\end{definition}

We start with the construction of a generating pair for a cyclic group.
\begin{lemma}
\label{lem-cyclic}
Let~$n$ be an integer greater than~$1$, and let~$m \eqdef \ceil{\sqrt{n} \,}$. There is a generating pair~$S_1, S_2$ for~$\integers_n$ such that~$\size{S_1} = m$ and~$\size{S_2} \le \floor{ n / m} + 1$. If~$n$ is a perfect square, the pair further satisfies~$\size{S_1} = \size{S_2} = \sqrt{n}\,$.
\end{lemma}
\begin{proof}
We use the Division Algorithm to find such subsets. If we choose a divisor~$m$ that roughly equals~$\sqrt{n}$, then the number of different remainders and quotients we get when we divide non-negative integers less than~$n$ are both roughly~$\sqrt{n}$.

Let~$m \eqdef \ceil{\sqrt{n}\, }$, let~$S_1 \eqdef \set{ 0, 1, 2, \dotsc, m - 1}$, and let~$S_2 \eqdef \set{ m i : 0 \le i < n/m }$. The subsets~$S_1, S_2$ satisfy the required properties.
\end{proof}

Next, we consider a direct product of two cyclic groups whose orders are both odd powers of the same prime number.
\begin{lemma}
\label{lem-odd-power}
Consider the group~$G \eqdef \integers_n \times \integers_m$, where~$n \eqdef p^k$, $m \eqdef p^l$, $p$ is prime, and~$k,l$ are positive odd integers. There is a generating pair~$S_1, S_2$ for~$G$ such that~$\size{S_1} = \size{S_2} = \sqrt{nm} \,$.
\end{lemma}
\begin{proof}
W.l.o.g., assume that~$k \ge l$. Let~$q \eqdef p^{(k + l)/2}$. Note that~$0 < q = \sqrt{nm} < n$. Consider
\begin{align*}
S_1 \quad & \eqdef \quad \integers_q \times \set{ 0 } \enspace, \quad \text{and} \\
S_2 \quad & \eqdef \quad \set{ iq : 0 \le i < n/q} \times \integers_m \enspace.
\end{align*}
By the Division Algorithm, we have
\[
\integers_q + \set{ iq : 0 \le i < n/q} \quad = \quad \integers_n \enspace.
\]
So~$S_1 + S_2 = G$. Moreover, we have~$\size{S_1} = q = \sqrt{nm} = (n/q) \times m = \size{S_2}$. 
\end{proof}

We may combine the generating pairs for two groups to obtain one for their direct product. While we present the proof of the lemma below with the notation for abelian groups, it also holds for non-abelian groups.
\begin{lemma}
\label{lem-product}
Consider a direct product group~$G \eqdef G_1 \times G_2$. Suppose~$S_1, S_2$ is a generating pair for~$G_1$, and~$T_1, T_2$ is a generating pair for~$G_2$. Then~$(S_1 \times T_1), (S_2 \times T_2)$ is a generating pair for~$G$.
\end{lemma}
\begin{proof}
Any element~$g \in G_1$ may be expressed as~$g_1 + g_2$, where~$g_1 \in S_1$ and~$g_2 \in S_2$. Similarly, any element~$h \in G_2$ may be expressed as~$h_1 + h_2$, where~$h_1 \in T_1$ and~$h_2 \in T_2$. Then~$(g, h) = (g_1, h_1) + (g_2, h_2) \in (S_1 \times T_1) + (S_2 \times T_2)$.
\end{proof}

The construction for a general finite abelian group combines the above building blocks.
\begin{theorem}
\label{thm-abelian}
For any finite abelian group~$G$ with order~$n$, there is a generating pair~$S_1, S_2$ for~$G$ such that~$\size{S_1}$ and~$\size{S_2}$ are both at most~$2\sqrt{n}$.
\end{theorem}
\begin{proof}
By the fundamental theorem of abelian groups, any non-trivial finite abelian group~$G$ may be expressed as a direct product of cyclic groups of prime power order:
\begin{equation}
\label{eq-decomp}
G \quad \cong \quad \integers_{p_1^{k_1}} \times \integers_{p_2^{k_2}} \times \dotsb \times \integers_{p_l^{k_l}} \enspace,
\end{equation}
where the integers~$p_i$ are primes, not necessarily distinct, and the integers~$k_i \ge 1$ for all~$i \in [l]$. We prove the lemma by strong induction on~$l$, the number of cyclic groups in a decomposition of~$G$.

If~$l = 1$, the statement follows from Lemma~\ref{lem-cyclic}.

Suppose the statement holds for all non-trivial finite abelian groups which have a decomposition as above with at most~$m$ cyclic groups, for some~$m \ge 1$. Suppose~$l \eqdef m+1$, and consider a group~$G$ with order~$n$ and a decomposition with~$l$ cyclic groups as in \cref{eq-decomp}. 

Suppose for some~$i \in [l]$, the integer~$k_i$ is even. Let~$r$ be such an index. We let~$G_1 \eqdef \integers_{p_r^{k_r}}$ and~$G_2$ the direct product of the remaining cyclic groups so that~$G \cong G_1 \times G_2$. By \Cref{lem-cyclic}, there is a generating pair~$S_1, S_2$ for~$G_1$ such that~$\size{S_1} = \size{S_2} = \sqrt{n_1}\,$, where~$n_1 \eqdef {p_r}^{k_r}$. By the induction hypothesis, there is a generating pair~$T_1, T_2$ for~$G_2$ such that~$\size{T_1}, \size{T_2} \le 2 \sqrt{n_2}\,$, where~$n_2 \eqdef \size{G_2}$. By \Cref{lem-product}, we get a generating pair~$(S_1 \times T_1), (S_2 \times T_2)$ for~$G$ such that~$ \size{S_1 \times T_1}, \size{S_2 \times T_2} \le 2 \sqrt{n_1 n_2} = 2 \sqrt{n}\,$.

Suppose for all~$i \in [l]$, the integers~$k_i$ are odd. Suppose for some~$i,j \in [l]$, $i \neq j$, we have~$p_i = p_j$. Let~$r,s$ be such a pair of indices. We let~$G_1 \eqdef \integers_{p_r^{k_r}} \times \integers_{p_s^{k_s}}$ and~$G_2$ the direct product of the remaining cyclic groups so that~$G \cong G_1 \times G_2$. By \Cref{lem-odd-power}, there is a generating pair~$S_1, S_2$ for~$G_1$ such that~$\size{S_1} = \size{S_2} = \sqrt{n_1}\,$, where~$n_1 \eqdef {p_r}^{k_r} {p_s}^{k_s}$. By the induction hypothesis, there is a generating pair~$T_1, T_2$ for~$G_2$ such that~$\size{T_1}, \size{T_2} \le 2 \sqrt{n_2}\,$, where~$n_2 \eqdef \size{G_2}$. By \Cref{lem-product}, we get a generating pair~$(S_1 \times T_1), (S_2 \times T_2)$ for~$G$ such that~$ \size{S_1 \times T_1}, \size{S_2 \times T_2} \le 2 \sqrt{n_1 n_2} = 2 \sqrt{n}\,$.

Suppose for all~$i \in [l]$, the integers~$k_i$ are odd, and the primes~$p_i$ are all distinct. Then, by the Chinese Remainder Theorem, we have~$G \cong Z_n$. The statement now follows from \Cref{lem-cyclic}, by noting that for~$n \ge 2$ we have~$\ceil{ \sqrt{n} \,} \le 2 \sqrt{n}$ and~$\floor{ \sqrt{n} \,} + 1 \le 2 \sqrt{n} \,$.
\end{proof}

The algorithm for the abelian Hidden Subgroup Problem follows directly from the existence of a suitable generating pair for the underlying group.
\begin{corollary}
\label{cor-ahsp}
There is a deterministic algorithm with query complexity at most~$4 \sqrt{n} \,$ that solves the Hidden Subgroup Problem over an abelian group~$G$ with order~$n$. Moreover, if the input function hides the non-trivial subgroup~$H$, the algorithm finds all the elements of~$H$. 
\end{corollary}
\begin{proof}
By \Cref{thm-abelian}, there is a generating pair~$S_1, S_2$ for the group~$G$ such that~$\size{S_1}, \size{S_2} \le 2 \sqrt{n}\, $. We query the oracle function~$f$ at~$-x$ for all~$x \in S_1$ and at all elements~$y \in S_2$. If the function is injective on the set of points queried, we say~$f$ is injective. Otherwise, we say that~$f$ is not injective, and output~$\set{ x + y : x \in S_1, ~ y \in S_2, ~ f(-x) = f(y) }$.

If the function is injective, the above algorithm outputs the correct answer. Suppose~$f$ hides the non-trivial subgroup~$H$. Let~$h \in H$ be any element of the hidden subgroup. We have~$h = h_1 + h_2 $ for some~$h_1 \in S_1$ and~$h_2 \in S_2$. By definition of~$H$ we have~$f(- h_1) = f( h_2)$. When~$h$ is not the identity, we have~$ - h_1 \neq h_2$ and the algorithm detects a collision. Moreover, the algorithm computes~$H$ correctly.
\end{proof}
Note that we can get upper bounds with better constants for certain abelian groups by appealing to the properties of the generating pair constructed in \Cref{thm-abelian}. As the Simon Problem is a special case, \Cref{thm-simon-det} implies that in the worst case, these upper bounds are tight up to a constant factor.

\section{The General Hidden Subgroup Problem}
\label{sec-general}

Finally, we consider the Hidden Subgroup Problem for a general finite group. Non-abelian groups are more varied in structure than abelian ones, and it appears challenging to extend the ideas used in the abelian case to them. Instead, we resort to a probabilistic argument to show the existence of a generating pair of small size. The construction comes within a poly-logarithmic multiplicative factor of optimal, and leads to our final algorithm.

\begin{theorem}
\label{thm-non-abelian}
For any finite group~$G$ with order~$n$ (with~$n > 1$), there is a generating pair~$S_1, S_2$ for~$G$ such that~$\size{S_1}$ and~$\size{S_2}$ are both at most~$\big\lceil { \sqrt{n \ln n} \,} \big\rceil$.
\end{theorem}
\begin{proof}
We let~$S_1 \eqdef \set{ g_1, g_2, \dotsc, g_t}$, any subset consisting of~$t$ distinct group elements, where~$ t \eqdef \big\lceil{ \sqrt{n \ln n} \,}\big\rceil$. Let~$\bR \eqdef \set{ \bh_1, \bh_2, \dotsc, \bh_t}$ be a set of~$t$ distinct group elements chosen uniformly at random from the collection of~$t$-element subsets of~$G$. Recall that the set of products of elements from~$S_1$ and~$\bR$ is denoted as~$S_1 \bR$, i.e., $S_1 \bR \eqdef \set{xy : x \in S_1, ~ y \in \bR}$.

We show that for any fixed element~$g \in G$, the probability that~$g \not\in S_1 \bR$ is less than~$1/n$.
\begin{align*}
\Pr( g \not\in S_1 \bR) \quad & = \quad \Pr( \forall i \in [t], ~ g_i^{-1} g \not\in \bR ) \\
    & = \quad \binom{n - t}{t} \binom{n}{t}^{-1} \\
    & = \quad \frac{(n - t) (n - t - 1) (n - t - 2) \dotsb (n - 2t + 1) }{ n (n - 1) (n - 2) \dotsb (n - t + 1)} \\
    & = \quad \left( 1 - \frac{t}{n} \right) \left( 1 - \frac{t}{n - 1} \right) \left( 1 - \frac{t}{n - 2} \right)        \dotsb \left( 1 - \frac{t}{n - t + 1} \right) \enspace. 
\end{align*}
Since~$t > 1$, we get
\[
\Pr( g \not\in S_1 \bR) 
    \quad  < \quad \left( 1 - \frac{t}{n} \right)^t
    \quad < \quad \exp \! \left( - \frac{t^2}{n} \right)
    \quad \le \quad \frac{1}{n} \enspace.
\]
By the Union Bound, the probability that~$S_1 \bR \not= G$ is strictly less than~$1$. So there is a subset~$S_2$ of size~$t$ such that~$S_1 S_2 = G$.
\end{proof}
A similar result may be derived from a generalization of the Erd{\H{o}}s-R{\'e}nyi Theorem~\cite{ER65-prob-method-group-theory} due to Babai and Erd{\H{o}}s~\cite{BE82-short-products}. However, this gives us a generating pair in which the size of a set may be as large as~$4 \sqrt{n \ln n} \,$.

As before, the algorithm for the general Hidden Subgroup Problem follows directly from the existence of a suitable generating pair for the underlying group.
\begin{corollary}
\label{cor-non-abelian-hsp}
There is a deterministic algorithm with query complexity at most~$2 \big\lceil \sqrt{n \ln n}\, \big\rceil$ that solves the Hidden Subgroup Problem over an arbitrary group~$G$ with order~$n$. Moreover, if the input function hides the non-trivial subgroup~$H$, the algorithm finds all the elements of~$H$.
\end{corollary}
\begin{proof}
By \Cref{thm-non-abelian}, there is a generating pair~$S_1, S_2$ for the group~$G$ such that~$\size{S_1}, \size{S_2} \le \big\lceil \sqrt{n \ln n}\, \big\rceil$. We query the oracle function~$f$ at~$x^{-1}$ for all~$x \in S_1$ and at all elements~$y \in S_2$. If the function is injective on the set of points queried, we say~$f$ is injective. Otherwise, we say that~$f$ is not injective, and output~$\set{ xy : x \in S_1, ~ y \in S_2, ~ f( x^{-1} ) = f(y) }$.

If the function is injective, the above algorithm outputs the correct answer. Suppose~$f$ hides the non-trivial subgroup~$H$. Let~$h \in H$ be any element of the hidden subgroup. We have~$h = h_1 h_2 $ for some~$h_1 \in S_1$ and~$h_2 \in S_2$. By definition of~$H$ we have~$f(h_1^{-1}) = f( h_2)$. When~$h$ is not the identity, we have~$h_1^{-1} \neq h_2$ and the algorithm detects a collision. Moreover, the algorithm computes~$H$ correctly.
\end{proof}

We can obtain explicit, optimal algorithms for certain classes of non-abelian groups. For example, if a group~$G$ of order~$n$ has a subgroup~$H$ of order~$\Theta( \sqrt{n} \,)$, then we can construct a generating pair~$S_1, S_2$ of size~$\Theta( \sqrt{n} \,) $ by taking~$S_1 \eqdef H$ and~$S_2$ to be a complete set of coset representatives of~$H$. In fact, in this case, the group satisfies a stronger property; see, for example, Ref.~\cite{BW20-cayley-sudoku}. Not all groups have a subgroup of such size. For example, the abelian group~$\integers_p$ for a prime~$p$ does not have such a subgroup, and yet admits a suitable generating pair.

\section{Subgroup-dependent query complexity}
\label{sec-subgroup-dependent}

There is a randomized algorithm for the Hidden Subgroup Problem that has expected query complexity~$ \Order(\sqrt{ \size{G} / \size{H}}\, )$ when the oracle~$f$ hides~$H$, and has worst-case query complexity~$\Order( \sqrt{ \size{G}} \,) $. In this section we use the notion of a generating pair in a more sophisticated manner to match this performance with a deterministic algorithm for all abelian groups, and some classes of non-abelian groups. 


The algorithm rests on the following observation.
\begin{lemma}
\label{lem-pigeon-hole}
Suppose~$G_1$ and~$H$ are subgroups of the possibly non-abelian finite group~$G$ such that~$G_1$ has cardinality at least~$\size{G}/ \size{H}$. Then either~$G_1 H = G$ or~$\size{G_1 \intersect H} > 1 $. 
\end{lemma}
\begin{proof}
The intersection~$G_1 \intersect H$ is a subgroup of~$G$ and contains the identity element. There are exactly~$\size{G}/ \size{H}$ distinct left cosets of~$H$. If~$G_1 H \neq G$, by the Pigeon-Hole Principle, there are two distinct elements~$g_1, g_2 \in G_1$ such that~$g_1 H = g_2 H$, i.e., $ g_2^{-1} g_1 \in H$. Since~$g_2^{-1} g_1$ is also an element of~$G_1$, and is not the identity element, we have~$\size{G_1 \intersect H} > 1$.
\end{proof}

Suppose that~$G$ is abelian, we know the order~$m$ of the hidden subgroup~$H$, and~$m > 1$. Then we may find a non-trivial element of~$H$ as follows. Consider any subgroup~$G_1$ of~$G$ of order~$n/m$, where~$n \eqdef \size{G}$. Such a subgroup exists since~$G$ has a subgroup with order~$d$ for any positive divisor~$d$ of~$n$~\cite[Corollary~2.4, page~77]{H74-algebra}. Let~$S_1, S_2$ be a generating pair for~$G_1$, and let~$g$ be any element of~$ G \setminus G_1$. By \Cref{lem-pigeon-hole}, either~(i) there is a non-identity element~$h \in H$ that is also in~$G_1$, or (ii)~$G_1 H = G$. In case~(i), let~$h = a + b$, where~$a \in S_1$ and~$b \in S_2$. If we query the oracle~$f$ at the elements~$ -x$ and~$y$, for all~$x \in S_1$ and~$y \in S_2$, we will find~$f( -a) = f(b)$, and can compute~$h$. In case~(ii), we have~$g = g_1 + h$ for some~$g_1 \in G_1$ and~$h \in H$. Since~$g$ is not in~$G_1$, the element~$h$ is not the identity. Suppose~$g_1 = a + b$, with~$a \in S_1$ and~$b \in S_2$. Then~$g - a = b + h$, and~$f(g - a) = f(b)$. If we query the oracle~$f$ at the elements~$g - x$ and~$y$, for all~$x \in S_1$ and~$y \in S_2$, we will find~$f(g - a) = f(b)$, and can compute~$h$. This is the key idea underlying the algorithm.

For sets~$S_1, S_2 \subseteq G$, define~$ S_1^{-1}$ as the set~$ S_1^{-1} \eqdef \set{ x^{-1} : x \in S_1 }$. Following our notation for the product of sets of group elements, $ S_1^{-1} S_2 = \set{ x^{-1} y : x \in S_1, ~ y \in S_2}$. Algorithm~\ref{alg-fc} (Find-Collision) implements the above idea with a geometrically decreasing sequence of guesses for the order of~$H$. We show in \Cref{thm-fc} that the algorithm is correct and has the query complexity we seek for a large class of groups.
\begin{algorithm}[ht]
\caption{Find-Collision($G, f$) \label{alg-fc}}

\SetKwInOut{Input}{Input}
\SetKwInOut{Output}{Output~}
\SetKwData{Inj}{injective}
\SetKwData{Coll}{collision}

\Input{group~$G$ of order~$n$, with~$n > 1$; oracle for~$f : G \rightarrow S$ that hides a subgroup}
\Output{\Inj, or \Coll~$a,b \in G$}

\BlankLine
Let~$k$ be the integer~$l$ such that~$n \in (2^l, 2^{l+1}]$ \;
\lIf{ $G$ is abelian }{ $k_0 \leftarrow 0$ }
\lElse{ $k_0 \leftarrow -1$ }
\While{ $k \ge k_0$ }{
  Find, if there is one, a subgroup~$G_1 \le G$ with the largest order in~$\big[ n/2^{k + 1}, ~ n/(\floor{2^k} + 1) \big]$ \tcc*[r]{the expression $\floor{2^k}$ is required to correctly handle the case~$k = k_0 = -1$}
  \uIf{ such a subgroup~$G_1$ exists}{
    Find a generating pair~$S_1, S_2$ for~$G_1$ which minimizes~$\max \set{ \size{S_1}, ~ \size{S_2} }$ \;
    \lIf{ $ G_1 = G$ }{ $ g \leftarrow e$, the identity element of~$G$ }
    \lElse{ $ g \leftarrow $ any element in~$ G \setminus G_1$ }
    Let~$R \leftarrow S_1^{-1} \union S_2 \union \big( S_1^{-1} \set{g} \big) $ \;
    Query~$f$ at all the elements in~$R$ \;
    \lIf{ $f(z) = f(y)$ for some~$z,y \in R$ such that~$z \neq y$}{\Return{\Coll~$z,y$} }
  }
  $k \leftarrow k - 1$ \;
}
\Return{\Inj}
\end{algorithm}

\begin{theorem}
\label{thm-fc}
Algorithm~\ref{alg-fc} (Find-Collision) solves the Hidden Subgroup Problem over any finite group~$G$. If the order of the group is~$n$ and that of the hidden subgroup~$H$ is~$m$, the algorithm has query complexity as stated below.
\begin{enumerate}
\item 
\label{part-abelian} 
If~$G$ is abelian, the algorithm has query complexity~$ \Order( \sqrt{n/m} \,)$.
\item 
\label{part-non-abelian} 
If~$G$ is not abelian, the algorithm has query complexity~$\Order( \sqrt{n \ln n} \,)$. Further, if~$G$ has a subgroup of order~$n_1$ such that~$n/m \le n_1 \le \kappa n/m$ for some~$\kappa  \ge 1$, then the query complexity is~$ \Order( \sqrt{(\kappa n/m) \ln (\kappa n/m)} \,)$.
\end{enumerate}
\end{theorem}
\begin{proof}
Since the algorithm outputs a collision only when it finds one, it gives the correct answer when the oracle function~$f$ is injective. Suppose the function~$f$ hides a non-trivial subgroup~$H$, so that~$m \ge 2$.

If the group is abelian, it has a subgroup of order~$n/m$~\cite[Corollary~2.4, page~77]{H74-algebra}, and~$n/m \le n/2$. If it does not find a collision in an earlier iteration, Algorithm~\ref{alg-fc} finds a collision in an iteration with~$k = \ell \ge 0$, where~$\ell$ is such that~$n/m \in \big[ n/2^{\ell + 1}, ~ n/(2^\ell + 1) \big]$. This is due to the reasoning given after \Cref{lem-pigeon-hole}. It thus outputs the correct answer. We have~$\ell = \ceil{ \log_2 m} - 1$, and initially, $k = \ceil{ \log_2 n} - 1 \ge 0$. By \Cref{thm-abelian}, we have~$\size{S_1}, \size{S_2} \le 2 \sqrt{ \size{G_1}}$ in every iteration with queries. So the query complexity of the algorithm is at most
\begin{align*}
\sum_{k = \ell}^{\ceil{ \log_2 n} - 1} 3 \cdot 2 \sqrt{ n/ 2^k} \quad
    & \le \quad 6 \sqrt{ \frac{n}{ 2^\ell}} ~ \sum_{i \ge 0} \frac{1}{ \sqrt{2^i}} \\
    & \le \quad 12 (1 + \sqrt{2} \,) \sqrt{n/m} \enspace,
\end{align*}
as~$\ell \ge \log_2 m - 1$. The bound on the query complexity when~$m = 1$ is the same as that for~$m = 2$, as the algorithm executes all the iterations until~$k = 0$. Part~\ref{part-abelian} of the theorem thus follows.

Suppose~$G$ is non-abelian. If the algorithm does not find a collision in earlier iterations, when~$k = -1$, we have~$G_1 = G$, and correctness follows as in \Cref{cor-non-abelian-hsp}. If~$G$ has a proper subgroup of order~$n_1$ with~$n/m \le n_1 \le \kappa n/m$ for some~$\kappa $, the algorithm finds a collision as in the abelian case in an iteration with~$k \ge \ell \ge 0$, where~$\ell$ is such that~$\kappa n/m \in \big[ n/2^{\ell + 1}, ~ n/(2^\ell + 1) \big]$. It thus outputs the correct answer. Further, we have~$\ell = \ceil{ \log_2 (m/\kappa )} - 1$, and initially, $k = \ceil{ \log_2 n} - 1 \ge 0$. By \Cref{thm-non-abelian}, we have~$\size{S_1}, \size{S_2} \le \sqrt{ \size{G_1} \ln \size{G_1}} + 1$ in every iteration with queries. So the query complexity of the algorithm is at most
\begin{align*}
\lefteqn{ \sum_{k = \ell}^{\ceil{ \log_2 n} - 1} 3 \left( 1 + \left( \frac{n}{ 2^k} \ln \frac{n}{ 2^k} \right)^{1/2} \right) } \\
    & \quad \le \quad 3( \ceil{ \log_2 n} - \ceil{ \log_2 (m/\kappa )} + 1 ) + 3 \left( \frac{n}{ 2^\ell} \ln \frac{n}{ 2^\ell} \right)^{1/2} \sum_{i \ge 0}  \frac{1}{ \sqrt{2^i}} \\
    & \quad \le \quad 6 + 3 \log_2 (\kappa n/m) + 3 (2 + \sqrt{2} \,) \left( \frac{2\kappa n}{m} \ln \frac{2\kappa n}{m} \right)^{1/2} \enspace,
\end{align*}
as~$\ell \ge \log_2 (m/\kappa ) - 1$. The bound on the query complexity when~$m = 1$ is~$2 \ceil{ \sqrt{ n \ln n} \,}$ more than that for~$m = 2$, as the algorithm executes all the iterations until~$k = -1$. Part~\ref{part-non-abelian} of the theorem thus follows.
\end{proof}
Unlike abelian groups, a non-abelian group of order~$n$ may not have a subgroup of order~$n/m$ when it has a proper subgroup of order~$m$. For example~$A_4$, the alternating group of degree~$4$, has order~$12$, has several subgroups of order~$2$, but does not have a subgroup of order~$6$. However, in large classes of instances of the Hidden Subgroup Problem, subgroups of suitable size exist. An immediate example is the class of CLT groups. (A group~$G$ is called a \emph{converse Lagrange Theorem (CLT) group\/} if it contains a subgroup of order~$d$ for every positive divisor~$d$ of~$\size{G}$.) CLT groups include supersolvable groups; see, e.g., Ref.~\cite{McLain57-subgroup-order}. For such instances, Find-Collision achieves query complexity~$\Order( \sqrt{ (n/m) \ln(n/m)} \,)$.

\section{Finding the hidden subgroup}
\label{sec-finding-subgroup}

Unlike the algorithms in \Cref{cor-ahsp} and \Cref{cor-non-abelian-hsp}, Algorithm~\ref{alg-fc} may find only one non-trivial element from the hidden subgroup. In this section, we show how to extend Algorithm~\ref{alg-fc} to find the entire hidden subgroup. 

When the group~$G$ is abelian, the problem of finding the entire subgroup may be reduced to that of finding one non-trivial element of the subgroup. This allows us to identify the subgroup by repeatedly using Algorithm~\ref{alg-fc} to find a set of generators.

\begin{algorithm}[ht]
\caption{Find-Abelian-Subgroup($G, f$) \label{alg-find-ahs}}

\SetKwInOut{Input}{Input}
\SetKwInOut{Output}{Output~}
\SetKwData{Inj}{injective}
\SetKwData{Coll}{collision}
\SetKwData{Gen}{generators}

\Input{group~$G$ of order~$n$, with~$n > 1$; oracle for~$f : G \rightarrow S$ that hides a subgroup}
\Output{\Inj, or \Gen~$S \subset G$ of the hidden subgroup}

\BlankLine
$S \leftarrow \emptyset$ \;
\Repeat{ outcome $=$ \Inj }{
  $H_1 \leftarrow \langle S \rangle$, the subgroup generated by~$S$ \;
  $G_1 \leftarrow G/H_1 $ \;
  \uIf{ $ \size{G_1} > 1$ }{
    Let~$f_1$ be the function defined by~$f$ and~$H_1$ in the proof of \Cref{thm-abelian-find-hs} \; 
    \textit{outcome\/} $\leftarrow$ Find-Collision($G_1, f_1$) \;
    \lIf{ outcome $=$ \Coll~$a,b$ }{ $S \leftarrow S \union \set{ g_1^{-1} g_2}$, where~$a = g_1 H_1$ and~$b = g_2 H_1$ }
  }
  \uElse{ \textit{outcome\/} $ \leftarrow$ \Inj }
}
\lIf{ $S = \emptyset$ }{ \Return{\Inj} }
\lElse{ \Return{ \Gen~$S$} }
\end{algorithm}

\begin{theorem}
\label{thm-abelian-find-hs}
There is a deterministic algorithm that solves the Hidden Subgroup Problem over any finite abelian group~$G$, and finds the hidden subgroup with~$ \Order( (\log m) \sqrt{n/m} \, )$ queries when the order of~$G$ is~$n$ and that of the hidden subgroup is~$m$.
\end{theorem}
\begin{proof}
Suppose the oracle is~$f$, the hidden subgroup is~$H$, and we know a set of generators for a subgroup~$H_1 \le H$.

Define a function~$f_1$ on the quotient group~$G/ H_1$ as~$f_1( g H_1) \eqdef f(g)$ for any~$g \in G$. The function~$f_1$ is well-defined as left cosets of~$H_1$ in~$G$ are subsets of left cosets of~$H$ in~$G$, and the function~$f$ is constant on left cosets of~$H$. Moreover, the function~$f_1$ hides the subgroup~$H/H_1$ of~$G/H_1$, as the left cosets of~$H$ in~$G$ correspond to left cosets of~$H/H_1$ in~$G/H_1$. Finally, the function~$f_1$ may be evaluated with one query to the oracle for~$f$.

Using this reduction, we may find the hidden subgroup~$H$ using Algorithm~\ref{alg-find-ahs}. The correctness of the algorithm follows by observing that in any iteration, if~$f_1$ is injective, then~$H_1$ equals the hidden subgroup. If~$f_1$ is not injective, by \Cref{thm-fc}, Find-Collision($ G_1, f_1$) returns a collision~$g_1 H_1, g_2 H_1 \in G/H_1$ for~$f_1$. Note that~$g_1, g_2$ is a collision for~$f$. We also have~$ g_1^{-1} g_2 \not\in H_1 $, so along with~$ g_1^{-1} g_2 $, the set~$S$ generates a larger subgroup of~$H$. Thus, the size of the subgroup~$H_1$ increases by a factor of at least~$2$ in every iteration a collision is found, and the algorithm terminates after at most~$\log_2 m$ iterations. In every iteration of Algorithm~\ref{alg-find-ahs}, the ratio of the order of~$G_1$ and the hidden subgroup~$H/ H_1$ equals~$n/m$. The query complexity of the algorithm now follows from \Cref{thm-fc}.
\end{proof}

It is not clear how to extend Algorithm~\ref{alg-find-ahs} to the non-abelian case, since the group $H_1$ may not be normal in general, and the corresponding quotient group may not be defined. We present a different algorithm, Algorithm~\ref{alg-find-hs}, that works for some abelian \emph{and\/} some non-abelian instances (which we describe after \Cref{thm-find-hs}). Algorithm~\ref{alg-find-hs} builds on Algorithm~\ref{alg-fnc}, which is a variant of Algorithm~\ref{alg-fc} and is also based on \Cref{lem-pigeon-hole}.
\begin{algorithm}[ht]
\caption{Find-New-Collision($G, H_1, f$) \label{alg-fnc}}

\SetKwInOut{Input}{Input}
\SetKwInOut{Output}{Output~}
\SetKwData{Nocoll}{no-new-collision}
\SetKwData{Coll}{collision}

\Input{group~$G$ of order~$n$, with~$n > 1$; subgroup~$H_1 \le G$; oracle for~$f : G \rightarrow S$ that hides a subgroup containing~$H_1$}
\Output{\Nocoll, or \Coll~$a,b \in G$ such that~$a^{-1}b \not\in H_1$}

\BlankLine
Let~$k$ be the integer~$l$ such that~$n \in (2^l, 2^{l+1}]$ \;
\lIf{ $G$ is abelian }{ $m_1 \leftarrow \max \set{ 2, \size{H_1} }$ }
\lElse{ $m_1 \leftarrow \size{H_1}$ }
$k_0 \leftarrow \ceil{ \log_2 m_1 } - 1$ \;
\While{ $k \ge k_0$ }{
  Find, if there is one, a subgroup~$G_1 \le G$ with the largest order in~$\big[ n/2^{k + 1}, ~ n/(\floor{2^k} + 1) \big]$ such that~$G_1 \intersect H_1 = \set{e}$, where~$e$ is the identity element of~$G$ \tcc*[r]{the expression $\floor{2^k}$ is required to correctly handle the case~$k = k_0 = -1$}
  \uIf{ such a subgroup~$G_1$ exists}{
    Find a generating pair~$S_1, S_2$ for~$G_1$ which minimizes~$\max \set{ \size{S_1}, ~ \size{S_2} }$ \;
    \lIf{ $ G_1 H_1 = G$ }{ $ g \leftarrow e$, the identity element of~$G$ }
    \lElse{ $ g \leftarrow $ any element in~$ G \setminus (G_1 H_1)$ }
    Let~$R \leftarrow S_1^{-1} \union S_2 \union \big( S_1^{-1} \set{g} \big) $ \;
    Query~$f$ at all the elements in~$R$ \;
    \lIf{ $f(z) = f(y)$ for some~$z,y \in R$ such that~$z^{-1}y \not\in H_1$}{\Return{\Coll~$z,y$} }
  }
  $k \leftarrow k - 1$ \;
}
\Return{\Nocoll}
\end{algorithm}
\begin{theorem}
\label{thm-fnc}
Let~$(G,f)$ be an instance of the Hidden Subgroup Problem, $H$ the subgroup that~$f$ hides, and~$H_1$ a subgroup of~$H$. Let the orders of~$G, H, H_1$ be~$n, m, m_1$, respectively. Suppose~$G$ has a subgroup~$G_0$ of order~$n_0$ such that~$n_0 \ge n/m $, and~$G_0$ intersects~$ H_1$ only in the identity element. Then Find-New-Collision($G, H_1, f$) (Algorithm~\ref{alg-fnc}) returns ``no new collision'' if~$H_1 = H$, and returns a collision~$a,b \in G$ such that~$a^{-1} b \not\in H_1$ otherwise. Further, if~$ n_0 \le \kappa n/m$ for some~$\kappa \ge 1$, the algorithm has query complexity as stated below:
\begin{enumerate}
\item 
\label{part-abelian-fnc} 
if~$G$ is abelian, the algorithm has query complexity~$ \Order( \sqrt{ \kappa n/m} \,)$, and
\item 
\label{part-non-abelian-fnc} 
if~$G$ is not abelian, the algorithm has query complexity~$\Order( \sqrt{( \kappa n/m) \ln ( \kappa n/m) } \,)$.
\end{enumerate}
\end{theorem}
\begin{proof}
When~$H_1 = \set{e}$, the algorithm is identical to Algorithm~\ref{alg-fc}, and its correctness follows by \Cref{thm-fc}. Suppose~$H_1$ is not the trivial subgroup.

Since the algorithm reports a collision~$a,b$ only when~$a^{-1} b \not\in H_1$, it gives the correct answer when~$H_1 = H$. Suppose~$H_1 \neq H$. By hypothesis, $G$ has a subgroup~$G_0$ of order at least~$n/m$ such that~$G_0$ intersects~$ H_1$ only in the identity element. Then~$\size{G_0} \le n/m_1 \le n/2$; otherwise, we would have two distinct elements of~$G_0$ in the same coset of~$H_1$, which implies that $\size{ G_0 \intersect H_1} > 1$.

If it does not find a collision in earlier iterations, the algorithm finds a subgroup~$G_1$ with~$\size{G_1} \ge n/m$ and~$G_1 \intersect H_1 = \set{e}$ in an iteration with~$k \ge k_0 \ge 0$ (recall that~$k_0 \eqdef \ceil{ \log_2 m_1 } - 1 \ge 0$, since~$m_1 \ge 2$).  Let~$S_1,S_2$ be the generating pair for~$G_1$ computed by the algorithm.

Suppose~$G_1 H_1 = G$. Consider any element~$h \in H \setminus H_1 $. We have~$h = g_1 h_1 $ for some~$g_1 \in G_1$ and~$h_1 \in H_1$. So~$g_1 = h h_1^{-1} \in G_1 \intersect H $, and~$g_1 \not\in H_1 $. We also have~$ g_1 = ab $ for some~$a \in S_1$ and~$b \in S_2$, so~$f(a^{-1}) = f(b)$. The algorithm queries the oracle~$f$ at the elements~$x^{-1}$ and~$y$, for all~$x \in S_1$ and~$y \in S_2$, so it finds and returns a collision.

Now suppose~$G_1 H_1 \neq G$, so that the algorithm also finds an element~$g \in G \setminus (G_1 H_1)$ in the same iteration. 

If~$G_1 H = G$, we have~$g = g_1 h$ for some~$g \in G_1$ and~$h \in H$. Since~$g \not\in G_1 H_1$, $g_1 h \not\in G_1 H_1$, and~$h \not\in H_1$. Suppose~$g_1 = ab$, with~$a \in S_1$ and~$b \in S_2$. Then~$a^{-1} g = b h$, and~$f(a^{-1} g) = f(b)$. The algorithm queries the oracle~$f$ at the elements~$x^{-1} g$ and~$y$, for all~$x \in S_1$ and~$y \in S_2$, so it finds~$f(a^{-1} g) = f(b)$, and returns a collision.

If~$G_1 H \neq G$, as in \Cref{lem-pigeon-hole}, there are two distinct elements~$g_1, g_2 \in G_1$ such that~$ g_2^{-1} g_1 \in H$. Since~$ G_1 \intersect H_1 = \set{e}$, and~$g_1 \neq g_2$, we have~$ g_2^{-1} g_1 \in H \setminus H_1$. We also have~$g_2^{-1} g_1 = ab$, for some~$a \in S_1$ and~$b \in S_2$. The algorithm queries the oracle~$f$ at the elements~$x^{-1}$ and~$y$, for all~$x \in S_1$ and~$y \in S_2$, so it finds~$f(a^{-1}) = f(b)$, and returns a collision in this case as well.

Assuming~$n/m \le \size{G_0} \le \kappa n/m \le n$ for some~$\kappa \ge 1$ with~$G_0$ as in the statement of the theorem, the algorithm executes iterations with~$k \ge \ell$, where~$\ell$ is such that~$\kappa n/m \in \big[ n/2^{\ell + 1}, ~ n/(\floor{2^\ell} + 1) \big]$. So the query complexity of the algorithm follows by same kind of analysis as in \Cref{thm-fc}.
\end{proof}

Thus, whenever there is a large enough subgroup~$G_0$ that intersects with a proper subgroup~$H_1$ of the hidden subgroup~$H$ only in the identity element, Algorithm~\ref{alg-fnc} gives us an element~$h$ of~$H$ that is not in~$H_1$. Along with~$H_1$, the element~$h$ generates a strictly larger subgroup of~$H$. As long as the condition above holds for all proper subgroups of~$H$, we can repeat Algorithm~\ref{alg-fnc} until we find a set of generators for the hidden subgroup. This process is described in Algorithm~\ref{alg-find-hs}.
\begin{algorithm}[ht]
\caption{Find-Subgroup($G, f$) \label{alg-find-hs}}

\SetKwInOut{Input}{Input}
\SetKwInOut{Output}{Output~}
\SetKwData{Inj}{injective}
\SetKwData{Coll}{collision}
\SetKwData{Gen}{generators}

\Input{group~$G$ of order~$n$, with~$n > 1$; oracle for~$f : G \rightarrow S$ that hides a subgroup}
\Output{\Inj, or \Gen~$S \subset G$ of the hidden subgroup}

\BlankLine
$S \leftarrow \emptyset$ \;
\Repeat{ outcome $=$ \Inj }{
  $H_1 \leftarrow \langle S \rangle$ \;
  \textit{outcome\/} $\leftarrow$ Find-New-Collision($G, H_1, f$) \;
  \lIf{ outcome $=$ \Coll~$a,b$ }{ $S \leftarrow S \union \set{ a^{-1} b}$ }
}
\lIf{ $S = \emptyset$ }{ \Return{\Inj} }
\lElse{ \Return{ \Gen~$S$} }
\end{algorithm}
\begin{theorem}
\label{thm-find-hs}
Algorithm~\ref{alg-find-hs} (Find-Subgroup) solves the Hidden Subgroup Problem over a group~$G$ and finds the hidden subgroup~$H$ when~$G$ has a subgroup~$G_0$ of order~$n/m$ such that~$\size{G_0 \intersect H} = 1$, where~$n$ and~$m$ are the orders of~$G$ and~$H$, respectively. Moreover, the algorithm makes
\begin{itemize}
\item $ \Order( (\log m) \sqrt{n/m} \, )$ queries when~$G$ is abelian, and
\item $ \Order( (\log m) \sqrt{(n/m) \log (n/m)} \, )$ queries when~$G$ is non-abelian.
\end{itemize}
\end{theorem}
\begin{proof}
The existence of a subgroup~$G_0$ as in the statement of the theorem implies that all the hypotheses of \Cref{thm-fnc} are satisfied for every subgroup~$H_1$ of~$H$. So, starting with~$H_1 = \set{e}$, where~$e$ is the identity element of~$G$, Find-New-Collision($G, H_1, f$) returns a collision~$a,b$ in each iteration of Algorithm~\ref{alg-find-hs} in which~$H_1 \neq H$. Since~$ a^{-1} b \not\in H_1 $, the set~$S \union \set{ a^{-1} b}$ generates a larger subgroup of~$H$. Thus, the size of the subgroup~$H_1$ increases by a factor of at least~$2$ in every iteration in which a collision is found, and the algorithm terminates after at most~$\log_2 m$ iterations. The query complexity follows from \Cref{thm-fnc}.
\end{proof}
A subgroup as in the statement of \Cref{thm-find-hs} exists if~$G$ is the semidirect product of~$H$ with another subgroup, i.e., $H$ is normal and there is a subgroup~$K$ such that~$G = H \rtimes K$, or there is a normal subgroup~$K$ such that~$G = H \ltimes K$. We may then take~$G_0 \eqdef K$. Not all groups have  a semidirect product structure, even if the hidden subgroup is normal. For example, every proper subgroup~$H$ of the abelian group~$\integers_{p^k}$ with~$k > 1$ is normal, but~$ \integers_{p^k}$ cannot be expressed as a semidirect product of~$H$ with another subgroup. On the other hand, a group need not have a semidirect product structure for a subgroup with the properties in \Cref{thm-find-hs} to exist. For example, consider~$S_n$, the symmetric group of degree~$n$, and~$S_{n-1}$ as its subgroup consisting of permutations that map~$n$ to itself. Then~$S_n = S_{n - 1} H$, where~$H$ is the subgroup generated by the cycle~$(1 ~ 2 ~ 3 ~ \dotsb ~ n)$, and neither~$S_{n - 1}$ nor~$H$ is normal in~$S_n$ for~$n \ge 4$. Such groups are known as the \emph{bicrossed\/} products (also as \emph{Zappa-Sz{\'e}p} or \emph{knit} products); see, e.g., \cite{ACIM09-bicrossed-product,Brin05-ZS-product}. Thus, Algorithm~\ref{alg-find-hs} finds the hidden subgroup~$H$ with query complexity as in \Cref{thm-find-hs} whenever~$G$ is the bicrossed product of~$H$ with another group. A description of groups arising as bicrossed products is a matter of ongoing research~\cite{ACIM09-bicrossed-product}.

\section{Open problems}
\label{sec-open}

We conclude with a few open problems. The query complexity of the algorithm designed by Ye and Li~\cite{YL21-hsp-classical} for finding the hidden subgroup in abelian instances may be smaller than that of Algorithm~\ref{alg-find-ahs}. The lower query complexity hinges on an intricate analysis of the structure of the hidden subgroup. Can we establish the same query complexity through simpler means? 

There are a number of variants of the Hidden Subgroup Problem, for example, when the underlying group is specified implicitly. These may also admit deterministic algorithms with optimal classical query complexity. The precise characterization of the deterministic query complexity of the Hidden Subgroup Problem for explicitly specified groups, especially in the non-abelian case, is perhaps the most interesting problem left open by this work. Related questions are whether there is a generating pair of size~$\Order( \sqrt{n} \,)$ for any non-abelian group of order~$n$, for what instances of the problem Algorithms~\ref{alg-fc} and~\ref{alg-fnc} give the correct output with query complexity~$\widetilde{\Order}( \sqrt{ n/m} \,)$, where~$m$ is the order of the hidden subgroup, or whether there are similar ``generic'' algorithms that achieve this query complexity for larger classes of groups.


\begin{thebibliography}{10}

\bibitem{ACIM09-bicrossed-product}
A.~L. Agore, A.~Chirv{\u{a}}situ, B.~Ion, and G.~Militaru.
\newblock Bicrossed products for finite groups.
\newblock {\em Algebras and Representation Theory}, 12(2):481--488, October 1,
  2009.

\bibitem{BE82-short-products}
L{\'a}szl{\'o} Babai and Paul Erd{\H{o}}s.
\newblock Representation of group elements as short products.
\newblock In Peter~L. Hammer, Alexander Rosa, Gert Sabidussi, and Jean Turgeon,
  editors, {\em Theory and Practice of Combinatorics}, volume~60 of {\em
  North-Holland Mathematics Studies}, pages 27--30. North-Holland, 1982.

\bibitem{BW20-cayley-sudoku}
Kady~Hossner Boden and Michael~B. Ward.
\newblock The heritage of {Cayley}-{Sudoku} tables.
\newblock Technical Report arXiv:2001.06711v1 [math.GR], arXiv.org, January
  2020.

\bibitem{Brin05-ZS-product}
Matthew~G. Brin.
\newblock On the {Zappa}-{Sz{\'e}p} product.
\newblock {\em Communications in Algebra}, 33(2):393--424, 2005.

\bibitem{CQ18-simon-problem}
Guangya Cai and Daowen Qiu.
\newblock Optimal separation in exact query complexities for {Simon's} problem.
\newblock {\em Journal of Computer and System Sciences}, 97:83--93, 2018.

\bibitem{CvD10-quantum-algorithms}
Andrew~M. Childs and Wim van Dam.
\newblock Quantum algorithms for algebraic problems.
\newblock {\em Reviews of Modern Physics}, 82:1--52, January 2010.

\bibitem{ER65-prob-method-group-theory}
Paul Erd{\H{o}}s and Alfr{\'e}d R{\'e}nyi.
\newblock Probabilistic methods in group theory.
\newblock {\em Journal d'Analyse Math{\'e}matique}, 14(1):127--138, December 1,
  1965.

\bibitem{H74-algebra}
Thomas~W. Hungerford.
\newblock {\em Algebra}, volume~73 of {\em Graduate Texts in Mathematics}.
\newblock Springer-Verlag, New York, 1974.

\bibitem{McLain57-subgroup-order}
D.~H. McLain.
\newblock The existence of subgroups of given order in finite groups.
\newblock {\em Mathematical Proceedings of the Cambridge Philosophical
  Society}, 53(2):278--285, 1957.

\bibitem{Nayak21-hsp-classical}
Ashwin Nayak.
\newblock Deterministic algorithms for the {Hidden Subgroup Problem}.
\newblock Technical Report arXiv:2104.14436v1 [cs.DS], arXiv.org, April 2021.

\bibitem{S97-simon-algorithm}
Daniel~R. Simon.
\newblock On the power of quantum computation.
\newblock {\em SIAM Journal on Computing}, 26(5):1474--1483, 1997.

\bibitem{WQTLC21-generalized-simon}
Zhenggang Wu, Daowen Qiu, Jiawei Tan, Hao Li, and Guangya Cai.
\newblock Quantum and classical query complexities for {Generalized Simon’s
  Problem}.
\newblock Technical Report arXiv:1905.08549v2 [quant-ph], arXiv.org, September
  2021.

\bibitem{YHLW21-generalized-simon}
Zekun Ye, Yunqi Huang, Lvzhou Li, and Yuyi Wang.
\newblock Query complexity of {Generalized Simon's Problem}.
\newblock {\em Information and Computation}, page 104790, 2021.

\bibitem{YL21-hsp-classical}
Zekun Ye and Lvzhou Li.
\newblock Deterministic algorithms for the hidden subgroup problem.
\newblock Technical Report arXiv:2110.00827v1 [cs.DS], arXiv.org, October 2021.

\end{thebibliography}

\end{document}